\documentclass[twoside]{article}
\usepackage[a4paper]{geometry}
\usepackage[latin1]{inputenc} 
\usepackage[T1]{fontenc} 
\usepackage{RR}
\usepackage{hyperref}

\usepackage{amssymb}
\usepackage{amsmath,amsthm}
\usepackage{url}
\usepackage{stmaryrd}
\usepackage{float} \usepackage{pgf} \usepackage{subfig} \usepackage{multirow}
\usepackage{xspace} \usepackage{hyperref} \usepackage{algorithm,algorithmic}
\usepackage{amsmath} 
\usepackage{calc} 

\usepackage{color}
\usepackage{bbm}

\newtheorem{theorem}{Theorem}
\newtheorem{definition}[theorem]{Definition}
\newtheorem{lemma}[theorem]{Lemma}
\newtheorem{lemma*}{Lemma}
{}
{}

\newtheorem{proposition}[theorem]{Proposition}

\newcommand {\C}   {\mathbb C}

\newcommand {\D}   {\mathbb D}
\newcommand {\F}   {\mathbb F}
\newcommand {\Z}   {\mathbb Z}

\newcommand {\Q}   {\mathbb Q}

\newcommand{\calT}   {\mathcal T}

\newcommand{\OO}{\ensuremath{{{O}}}}

\newcommand{\sO}{\ensuremath{\widetilde{{O}}}}
\newcommand{\sOB}{\ensuremath{\widetilde{{O}}_B}}

\newcommand{\makeremark}[2]{
  \newcommand{#1}[1]
    {
    \color{blue}
     $\longrightarrow$ \textsc{#2: }
     ##1
     $\longleftarrow$
    \color{black}
    }
}    

\makeremark{\SL}{SL}
\makeremark{\FR}{Fabrice says}
\makeremark{\MP}{Marc says}
\makeremark{\ET}{Elias says}
\makeremark{\LP}{Luis shouts}
\makeremark{\YB}{Yacine says}

\definecolor{1ST}{rgb}{1,0,0}
\definecolor{2ND}{rgb}{1,0.5,0}
\definecolor{3RD}{rgb}{1,0,1}

\newcommand{\ideal}[1]{\langle #1 \rangle }

\newcommand{\shutup}[1]{}

\renewcommand{\leq}{\leqslant}  
\renewcommand{\geq}{\geqslant}

\makeatletter
\def\cramped                           
 {\parskip0pt\@topsep0pt       
  \itemsep0pt\parsep0pt
}        
\makeatother

\usepackage{stmaryrd}

\usepackage{verbatim}
\newcommand {\comp}{\ensuremath{\sO_B(d^8+d^7\tau)}} 

\RRNo{8261}
\RRdate{March 2013}
\RRauthor{
Yacine Bouzidi\thanks{INRIA Nancy Grand Est, LORIA laboratory, Nancy, France. {\tt Firstname.Name@inria.fr}}
        \and
        Sylvain Lazard\footnotemark[1] 
        \and
        Marc Pouget\footnotemark[1] 
        \and
        Fabrice Rouillier\thanks{INRIA Paris-Rocquencourt and IMJ
          (Institut de Math\'ematiques de Jussieu, Universit\'e Paris 6, CNRS), Paris, France. {\tt Firstname.Name@inria.fr}}
}
\authorhead{Bouzidi \& Lazard \& Pouget \& Rouillier}
\RRtitle{Forme lin\'eaire s\'eparante de syst\`emes bivari\'es}
\RRetitle{Separating linear forms for bivariate systems}
\titlehead{Separating linear forms for bivariate systems}
\RRresume{Nous pr\'esentons un algorithme pour calculer une forme lin\'eaire
  s\'eparante d'un syst\`eme de polyn\^omes \`a deux variables \`a coefficients
  entiers, c'est-\`a-dire une combinaison lin\'eaire des variables qui prend des
  valeurs diff\'erentes quand elle est \'evalu\'ee en des solutions (complexes)
  distinctes du syst\`eme. En d'autres termes, une forme lin\'eaire s\'eparante
  d\'efinit un changement de coordonn\'ees qui met le syst\`eme alg\'ebrique en position
  g\'en\'erique, au sens o\`u deux solutions distinctes ne sont jamais
  verticalement align\'ees. Le calcul de ces formes lin\'eaires est au coeur de la
  plupart des algorithmes qui permettent de r\'esoudre des syst\`emes alg\'ebriques au
  moyen de param\'etrisations rationnelles des solutions et, de plus, le calcul
  d'une forme lin\'eaire s\'eparante domine la complexit\'e binaire de ces
  algorithmes.

  Etant donn\'es deux polyn\^omes \`a deux variables de degr\'e total au plus $ d $ avec
  des coefficients entiers de taille binaire au plus $ \tau $, notre algorithme
  calcule une forme lin\'eaire s\'eparante en \comp\ op\'erations binaires dans le pire
  des cas, am\'eliorant la meilleure complexit\'e connue pour ce probl\`eme d'un facteur
  $d^2$ (o\`u $ \sO $ se r\'ef\`ere \`a la complexit\'e o\`u les facteurs
  polylogarithmiques  sont omis et $O_B$ se r\'ef\`ere \`a la complexit\'e binaire).
}

 \RRabstract{We present an algorithm for computing a separating linear form of
  a system of bivariate polynomials with integer coefficients, that is a linear
  combination of the variables that takes different values when evaluated at
  distinct (complex) solutions of the system. In other words, a separating
  linear form defines a shear of the coordinate system that sends the algebraic
  system in generic position, in the sense that no two distinct solutions are
  vertically aligned. The computation of such linear forms is at the core of
  most algorithms that solve algebraic systems by computing rational
  parameterizations of the solutions and, moreover, the computation a separating
  linear form is the bottleneck of these algorithms, in terms of worst-case bit
  complexity.

  Given two bivariate polynomials of total degree at most $d$ with integer
  coefficients of bitsize at most~$\tau$, our algorithm computes a separating
  linear form in \comp\  bit operations in the worst case,
  where the previously known best bit complexity for this problem was
  $\sOB(d^{10}+d^9\tau)$ (where $\sO$ refers to the complexity where
  polylogarithmic factors are omitted and $O_B$ refers to the bit complexity). }
\RRmotcle{calcul formel,  r\'esolution de syst\`emes polynomiaux, forme lin\'eaire
  s\'eparante}
\RRkeyword{computer algebra, polynomial system solving, linear separating form}
 \RRprojet{Vegas}  
 \URLorraine 
 \RCNancy 

\begin{document}
\makeRR  

\section{Introduction}\label{sec:intro}
One approach, that can be traced back to Kronecker, to solve a system of
polynomials with a finite number of solutions is to compute a rational
parameterization of its solutions. 
Such a representation of the (complex) solutions of a system is given by a set
of univariate polynomials and associated rational one-to-one mappings that send
the roots of the univariate polynomials to the solutions of the system.
%
%
Such parameterizations enable to reduce computations on the system to
computations with univariate polynomials and thus ease, for instance, the
isolation of the solutions or the evaluation of other polynomials at the
solutions.

The computation of such parameterizations has been a focus of interest for a
long time; see for example
\cite{ABRW,VegKah:curve2d:96,Rou99,GLS01,bostan2003fast,det-jsc-2009} and
references therein. Most algorithms first shear the coordinate system, with a
linear change of variables, so that the input algebraic system is in generic
position, that is such that no two solutions are vertically aligned. These
algorithms thus need a \emph{linear separating form},
that is a linear combination of the coordinates that takes different values when
evaluated at different solutions of the system. Since a random linear form is
separating with probability one, 
probabilist Monte-Carlo algorithms can overlook this issue.
However, for deterministic algorithms, computing a linear separating
  form is critical, especially because this is, surprisingly, the current
  bottleneck for bivariate systems, as discussed below.

We restrict our attention to systems of two bivariate polynomials of
  total degree bounded by $d$ with integer coefficients of bitsize bounded by
  $\tau$. For such systems, the approach with best known worst-case bit
  complexity for computing a rational parameterization was first introduced by
  Gonzalez-Vega and El Kahoui \cite{VegKah:curve2d:96} (see also
  \cite{lgv-in-etdidapc-02}): their initial analysis of $\sOB(d^{16}+d^{14}\tau
  ^2)$ was improved by Diochnos et al.  \cite[Lemma 16 \& Theorem
  19]{det-jsc-2009}\footnote{The overall bit complexity stated in \cite[Theorem
    19]{det-jsc-2009} is $\sOB(d^{12}+d^{10}{\tau}^2)$ because it includes the
    isolation of the solutions of the system. Note that this complexity
    trivially decreases to $\sOB(d^{10}+d^9\tau)$ by the recent result of
    Sagraloff \cite{sagraloff2012issacNewDsc} which improves the complexity of
    isolating the real roots of a univariate polynomial.
    Note also that Diochnos et al.  \cite{det-jsc-2009} present two algorithms,
    the M\_RUR and G\_RUR algorithms, both with bit complexity
    $\sOB(d^{12}+d^{10}{\tau}^2)$. However, this complexity is worst case only
    for the M\_RUR algorithm. As pointed out by Emeliyanenko and Sagraloff
    \cite{sagraloff2012issacBisolve}, the G\_RUR algorithm uses a modular gcd
    algorithm over an extension field whose considered bit complexity is
    expected.} to (i) $\sOB(d^{10}+d^9\tau)$ for computing a separating linear
  form and then (ii) $\sOB(d^{7}+d^6\tau)$ for computing a
  parameterization. Computing a separating linear form is thus the bottleneck of
  the computation of the rational parameterization.  This is still true even
  when considering the additional phase of computing isolating boxes of the
  solutions (from the rational parameterization), which state-of-the-art
  complexity is in 
  {$\sOB(d^8+d^7\tau)$} \cite[Proposition
  19]{bouzidi2013RurWorstCase}. 

\paragraph{Main results. }
Our main contribution is a new deterministic algorithm of worst-case bit
complexity \comp\ for computing a separating linear form of a system of two
bivariate polynomials of total degree at most $d$ and integer coefficients of
bitsize at most $\tau$ (Theorem~\ref{th:algo:sep-elem}). {This decreases
  by a factor $d^2$ the best known complexity for this problem.}

As a direct consequence, using our algorithm for computing a separating linear
form directly yields a rational parameterization within the same overall
complexity as our algorithm, both in the approach of Gonzalez-Vega et
al. \cite{VegKah:curve2d:96,det-jsc-2009} and in that of Bouzidi et
al. \cite{bouzidi2013RurWorstCase} for computing the alternative rational
parameterization as defined in \cite{Rou99}.  As a byproduct, we obtain an
algorithm for computing the number of {(complex)} distinct solutions of such systems within
the same complexity, i.e. \comp.  


\section{Overview and  organization}
\label{sec:overview}

Let $P$ and $Q$ be two bivariate polynomials of total degree bounded by $d$ and
integer coefficients of maximum bitsize $\tau$. Let $I=\ideal{P,Q}$ be the ideal
they define and suppose that $I$ is zero-dimensional.  The goal is to find a linear form $T=X+aY$, with $a\in\mathbb{Z}$,
that separates the solutions of $I$.

We first outline a classical algorithm which is essentially the same as those
proposed, for instance, in \cite[Lemma 16]{det-jsc-2009} and \cite[Theorem
24]{kerber11JSC}\footnote{The stated complexity of \cite[Theorem
  24]{kerber11JSC} is $\sOB(d^9\tau)$, but it seems the fact that the sheared
  polynomials have bitsize in $\sO(d+\tau)$ (see
  Lemma~\ref{lem:complexity:shear}) instead of $\sO(\tau)$ has been overlooked
  in their proof.} and whose complexity, in $\sOB(d^{10}+d^9\tau)$, is the best
known so far for this problem. This algorithm serves two purposes: it gives some
insight on the more involved \comp-time algorithm that
follows and it will be used in that algorithm but over $\Z/\mu\Z$ instead
of~$\mathbb{Z}$.

\paragraph{Known $\mathbf{\sOB(d^{10}+d^9\tau)}$-time algorithm for computing a
  separating linear form.}

The idea is to work with a ``generic" linear form $T=X+SY$, where $S$ is an
indeterminate, and find conditions such that the specialization of $S$ by an
integer $a$ gives a separating form.  We thus consider $P(T-SY,Y)$ and
$Q(T-SY,Y)$, the ``generic'' sheared polynomials associated to $P$ and $Q$, and
$R(T,S)$ their resultant with respect to $Y$.  This polynomial has been
extensively used and defined in several context; see for instance the related
$u$-resultant  \cite{VanDerWaerden}.

It is known that, in a set $\mathcal{S}$ of $d^4$ integers, there exists at
least one integer $a$ such that $X+aY$ is a separating form for $I$ since $I$
has at most $d^2$ solutions which define at most $d^2\choose 2$ directions in
which two solutions are aligned.  Hence, a separating form can be found by
computing, for every $a$ in $\cal S$, the degree of the squarefree part of
${R(T,a)}$ and by choosing 
 one $a$ for which this degree is maximum.  Indeed, for
any (possibly non-separating) linear form $X+aY$, the number of distinct roots
of $R(T,a)$, which is the degree of its squarefree part, is always smaller than
or equal to the number of distinct solutions of $I$, and equality is attained
when the linear form $X+aY$ is separating (Lemma~\ref{lem:ineq1}). The
complexity of this algorithm is in $\sOB(d^{10}+d^9\tau)$ because, for $d^4$
values of $a$, the polynomial $R(T,a)$ can be shown to be of degree $O(d^2)$ and
bitsize $\sO(d^2+d\tau)$, and its squarefree part can be computed in
$\sOB(d^6+d^5\tau)$ time.

\paragraph{$\mathbf{\sO_B(d^8+d^7\tau)}$-time algorithm for computing a separating linear form.}
To reduce the complexity of the search for a separating form, one can first
consider to perform naively the above algorithm on the system $I_\mu=\ideal{P \bmod \mu, Q \bmod
  \mu}$ in $\Z_\mu=\Z/\mu\Z$, where
$\mu$ is a prime number upper bounded by some polynomial in $d$ and $\tau$ (so that the bit 
complexity of  arithmetic
operations in $\Z_\mu$  is polylogarithmic in $d$ and $\tau$).
%
The resultant $R_\mu(T,S)$ of ${P}(X-SY,Y)\bmod \mu$ and ${Q}(X-SY,Y)\bmod \mu$
with respect to $Y$ can be computed in $\sOB(d^6+d^5\tau)$ 
bit operations and, since its degree is at most $2d^2$ in each variable,
evaluating it at $S=a$ in $\Z_\mu$ can be easily done in $\sOB(d^4)$ bit
operations.
Then, the computation of its squarefree part does not
suffer anymore from the coefficient growth, and it becomes softly linear in its
degree, that is $\sOB(d^2)$. Considering $d^4$ choices of $a$, we get an
algorithm that computes a separating form for $I_\mu$ in $\sOB(d^8)$ time in
$\Z_\mu$. However, a serious problem remains, that is to ensure that a
separating form for $I_\mu$ is also a separating form for $I$. This issue
requires to develop a more subtle algorithm.

We first show, in Section~\ref{sec:sep-modulo}, a critical property (Proposition
\ref{prop:separation-correspondence}) which states that a separating linear form
over $\Z_\mu$ is also separating over $\Z$ when $\mu$ is a \emph{lucky} prime
number, which is, essentially, a prime such that the number of solutions of
$\ideal{P,Q}$ is the same over $\Z$ and over $\Z_\mu$. We then show in
Sections~\ref{sec:lucky-nbroot} to \ref{sec:gather-lucky} how to compute such a
lucky prime number. We do that by first proving in Section
\ref{sec:lucky-nbroot} that, under mild conditions on $\mu$, the number of
solutions over $\Z_\mu$ is always less than or equal to the number of solutions
over $\Z$ (Proposition~\ref{prop:lessroot}) and then by computing a bound on the
number of unlucky primes (Proposition~\ref{prop:luckyroot}). Computing a lucky
prime can then be done by choosing a $\mu$ that maximizes the number of
solutions over $\Z_\mu$ among a
set of primes of cardinality $\widetilde{\Theta}(d^4+d^3\tau)$.  For that
purpose, we present in Section \ref{sec:counting-Zmu} a new algorithm, of
independent interest, for computing in $\sO(d^4)$ arithmetic operations the
number of distinct solutions of the system $I_\mu$ in $\Z_\mu$; this algorithm
is based on a classical triangular decomposition.  This yields, in
Section~\ref{sec:gather-lucky}, a \comp-time algorithm for computing a lucky
prime $\mu$ in $\sO(d^4+d^3\tau)$.  Now, $\mu$ is fixed, and we can apply the
algorithm outlined above for computing a separating form for $I_\mu$ in $\Z_\mu$
in $\sOB(d^8)$ time (Section~\ref{sec:gathering}). This form, which is also
separating for $I$, is thus obtained with a total bit complexity of \comp\
(Theorem~\ref{th:algo:sep-elem}).

\section{Notation and preliminaries}\label{sec:prelim}
We introduce notation and recall classical material about subresultant
sequences.

The bitsize of an integer $p$ is the number of bits needed to represent it, that
is $\lfloor\log p\rfloor+1$ ($\log$ refers to the logarithm in base 2). For
rational numbers, we refer to the bitsize as to the maximum bitsize of its
numerator and denominator.  The bitsize of a polynomial with integer or rational
coefficients is the \emph{maximum} bitsize of its coefficients. As mentioned
earlier, $O_B$ refers to the bit complexity and $\sO$ and $\sOB$ refer to complexities
 where polylogarithmic factors are omitted.

In the following, $\mu$ is a prime number and we denote by $\Z_\mu$ the quotient
$\Z/\mu\Z$.  We denote by $\phi_\mu$: $\Z \rightarrow \Z_\mu$ the reduction
modulo $\mu$, and extend this definition to the reduction of polynomials with
integer coefficients.  We denote by $\D$ a unique factorization domain,
typically $\Z[X,Y]$, $\Z[X]$, $\Z_\mu[X]$, $\Z$ or $\Z_\mu$. We also denote by
$\F$ a field, typically $\Q$, $\C$, or $\Z_\mu$.

For any polynomial $P\in \D[X]$, let $Lc_X(P)$ denote its leading coefficient
with respect to the variable $X$, 
$d_X(P)$ its degree with respect to $X$, and $\overline{P}$ its squarefree
part. The ideal generated by two polynomials $P$ and $Q$ is denoted
$\ideal{P,Q}$, and the affine variety of an ideal $I$ is denoted by $V(I)$; in
other words, $V(I)$ is the set of distinct solutions of the system $\{P,Q\}$.
The solutions are always considered in the algebraic closure of $\D$
and the number of distinct solutions is denoted by $\#V(I)$.
For a point $\sigma \in V(I)$, $\mu_I(\sigma)$ denotes the
multiplicity of $\sigma$ in $I$. For simplicity, we refer indifferently to the
ideal $\ideal{P,Q}$ and to the  system $\{P,Q\}$.

\emph{We finally introduce the following notation which are extensively used
  throughout the paper.} Given the two input polynomials $P$ and $Q$, we
consider the ``generic'' change of variables $X=T-SY$, and define the
``sheared'' polynomials $P(T-SY,Y)$, $Q(T-SY,Y)$, and their resultant with
respect to $Y$,
\begin{equation}\label{eq0}
{ R(T,S)}=Res_Y({P}(T-SY,Y),{Q}(T-SY,Y)).
\end{equation}
The complexity bounds on the degree, bitsize and computation of these
polynomials are analyzed at the end of this section in
Lemma~\ref{lem:complexity:shear}.  Let $L_R(S)$ be the leading coefficient of
$R(T,S)$ seen as a polynomial in $T$.  Let ${ L_{P}(S)}$ and ${ L_{Q}(S)}$ be
the leading coefficients of ${P}(T-SY,Y)$ and $Q(T-SY,Y)$, seen as polynomials
in $Y$
; it is straightforward that these leading coefficients do not depend on $T$. In
other words:

\begin{equation}\label{eq1}
\begin{array}{lll}
L_{P}(S) = Lc_Y({P}(T-SY,Y)),\ \ & L_{Q}(S) = Lc_Y({Q}(T-SY,Y)), \ \ & L_R(S)= Lc_T(R(T,S)).
\end{array}
\end{equation}

%

\subsection{Subresultant sequences}
\label{sec:subres-prelim}

We recall here the definition of subresultant sequences and some related
properties. Note that we only use subresultants in
Section~\ref{sec:tri-dec} in which we recall a classical triangular
decomposition algorithm. 

We first recall the concept of
\emph{polynomial determinant} of a matrix which is used in the definition of subresultants.
Let $M$ be an $m \times n$ matrix with $m \leq n$ and $M_i$ be the square submatrix of $M$ consisting of the first $m-1$ columns
and the $i$-th column of $M$, for $i=m,\ldots,n$. The \emph{polynomial determinant} of $M$ is the polynomial
defined as $\det(M_{m})Y^{n-m}+\det(M_{m + 1})Y^{n-(m+1)}+\ldots + \det(M_{n})$.

Let $P =\sum_{i=0}^p a_i Y^i$ and $Q= \sum_{i = 0}^q b_i Y^i$ be two polynomials
in $\D[Y]$ and assume without loss of generality that $p \geq q$. The Sylvester
matrix of $P$ and $Q$, $Sylv(P,Q)$ is the $(p+q)$-square matrix whose rows are
$Y^{q-1}P,\ldots,P,Y^{p-1}Q,\ldots,Q$ considered as vectors in the basis
$Y^{p+q-1},\ldots,Y,1$. 

\begin{center}

\begin{small}
$Sylv(P,Q)=\begin{array}{cc} 
\overbrace{\rule{5.5cm}{0pt}}^{\text{p+q columns}} & \\
\begin{pmatrix}
 a_p & a_{p-1} & \cdots & \cdots & a_0 &  & &  \\
&\hspace{-,3cm}a_p & \hspace{-,6cm}a_{p-1} &  \cdots & \cdots & a_0 & &  \\
& & \hspace{-,7cm}\ddots & & & & \hspace{-,2cm}\ddots &  \\
&  &  & \hspace{-,9cm}a_p & \hspace{-,7cm}a_{p-1} &  \cdots & \cdots  & a_0 \\
b_q & b_{q-1} & \cdots & & \hspace{-,9cm}b_0 &  &  \\
&\hspace{-,3cm}b_q & \hspace{-,4cm}b_{q-1} & \cdots & \hspace{2mm}b_0 & &  &  \\
& & \hspace{-,5cm}\ddots & & & \hspace{-,2cm}\ddots & & \\
&  & & \hspace{-,7cm}\ddots &  & & \ddots & \\
& & & & \hspace{-,9cm}b_q & \hspace{-,6cm}b_{q-1} &\hdotsfor{1} & \hspace{2mm}b_0
\end{pmatrix}
&
 \hspace{-,4cm}\begin{array}{c} 
      \left. \rule{0pt}{0.9cm} \right\} \text{\footnotesize q rows}  \\
     \left. \rule{0pt}{1.15cm} \right\} \text{\footnotesize p rows}  
   \end{array} 
\end{array}$
\end{small}
\end{center}

\begin{definition}\label{def:sub-resultant}(\cite[\S 3]{Kahoui03}). 
  For $i=0,\ldots, \min(q,p-1)$, let $Sylv_i(P,Q)$ be the $(p+q-2i)\times
  (p+q-i)$ matrix obtained from $Sylv(P,Q)$ by deleting the $i$ last rows of the
  coefficients of $P$, the $i$ last rows of the coefficients of $Q$, and the $i$
  last columns.
 
  For $i=0,\ldots, \min(q,p-1)$, the $i$-th polynomial subresultant of $P$ and
  $Q$, denoted by $Sres_{Y,i}(P,Q)$ is the polynomial determinant of
  $Sylv_i(P,Q)$.  When $q=p$, the $q$-th polynomial subresultant of $P$ and $Q$
  is $b_q^{-1}Q$.%
  \footnote{It can be observed that, when $p>q$, the $q$-th subresultant is
    equal to $b_q^{p-q-1}Q$, however it is not defined when $p=q$.  In this
    case, following El Kahoui, we extend the definition to $b_q^{-1}Q$ assuming
    that the domain $\D$ is integral, which is the case in this paper.  Note
    that it is important to define the $q$-th subresultant to be a multiple of
    $Q$ so that Lemma~\ref{lem:fund-prop-subres} holds when $Q(\alpha,Y)$ is of
    degree $q$ and divides $P(\alpha,Y)$ for some $\alpha$.}
\end{definition}

$Sres_{Y,i}(P,Q)$ has degree at most $i$ in $Y$, and the coefficient of its
monomial of degree $i$ in $Y$, denoted by ${sres}_{Y,i} (P, Q)$, is called the
$i$-th \emph{principal subresultant coefficient}.  Note that ${Sres}_{Y,0} (P,
Q)={sres}_{Y,0} (P, Q)$ is the \emph{resultant} of $P$ and $Q$ with respect to
$Y$, which we also denote by $Res_Y(P,Q)$.  Furthermore, the first (with respect
to increasing $i$) nonzero subresultant of $P,Q \in \mathbb{D}[Y]$ is equal to
their gcd in $\F_\mathbb{D}[Y]$, up to a multiplicative factor in
$\F_\mathbb{D}$, where $\F_\mathbb{D}$ is the fraction field of $\D$ (e.g., if
$\D=\Z[X]$, then $\F_\mathbb{D}=\Q(X)$, the field of fractions of polynomials in
$\Q[X]$);
more generally, the subresultants of $P$ and $Q$ are equal to either $0$ or to
polynomials in the remainder sequence of $P$ and $Q$ in Euclid's algorithm (up
to multiplicative factors in $\D$) \cite[\S8.3.3 \&
Cor. 8.32]{BPR06}.\footnote{For efficiency, the computation of subresultant
  sequences are usually performed by computing the polynomial remainder
  sequences using some variants of Euclid algorithm instead of the
  aforementioned determinants.}

We state below a fundamental property of subresultants which is instrumental in
the triangular decomposition algorithm used in Section~\ref{sec:tri-dec}. For
clarity, we state this property for bivariate polynomials $P =\sum_{i=0}^p a_i
Y^i$ and $Q= \sum_{i = 0}^q b_i Y^i$ in $\D[X,Y]$, with $p\geq q$.
Note that this property is often stated with a stronger assumption that is that
\emph{none} of the leading terms $a_p(\alpha)$ and $b_q(\alpha)$ vanishes.  This
property is a direct consequence of the specialization property of subresultants
and of the gap structure theorem; see for instance \cite[Lemmas 2.3, 3.1 and
Corollary 5.1]{Kahoui03}.

\begin{lemma}\label{lem:fund-prop-subres}
%
  For any $\alpha$ such that $a_p(\alpha)$ and $b_q(\alpha)$ do not both vanish,
  the first ${Sres}_{Y,k}(P,Q)(\alpha,Y)$ (for $k$ increasing) that does not
  identically vanish is of degree $k$ and it is the gcd of $P(\alpha,Y)$ and
  $Q(\alpha,Y)$ (up to a nonzero constant in the fraction field of
  $\D(\alpha)$).
\end{lemma}

\subsection{Complexity}
We recall complexity results, using fast algorithms, on subresultants and gcd
computations.  We also analyze complexities related to the computation of the
``sheared'' polynomials and their resultant.

\begin{lemma}[{\cite[Proposition 8.46]{BPR06} \cite[\S 8, 
Algorithm 7.3]{Reischert1997}}]
  \label{complexity:subresultant}
  Let $P$ and $Q$ in $\mathbb{Z}[X_1,\ldots, X_n][Y] $ of
  coefficient bitsize $\tau$ such that their degrees in $Y$ are bounded by
  $d_{Y}$ and their degrees in the other variables are bounded by $d$. 
  \begin{itemize}\cramped
  \item The coefficients of $Sres_{Y,i}(P,Q)$ have bitsize in
    $\sO(d_{Y}\tau)$. 
  \item The degree in $X_j$ of $Sres_{Y,i}(P,Q)$ is at most
    $2d(d_{Y}-i)$. 
  \item Any subresultants $Sres_{Y,i}(P,Q)$ can be computed in
    $\sO(d^{n} d_{Y}^{n+1})$ arithmetic operations, 
    and $\sOB(d^{n}
    d_{Y}^{n+2}\tau)$ bit operations. 
\end{itemize}
\end{lemma}    

In the sequel, we often consider the gcd of two univariate polynomials $P$ and
$Q$ and the gcd-free part of $P$ with respect to $Q$, that is, the divisor $D$
of $P$ such that $P=\gcd(P,Q)D$. Note that when $Q=P'$, the latter is the
squarefree part $\overline{P}$.

\begin{lemma}[{\cite[Remark 10.19]{BPR06}}]
\label{complexity:gcd}
  
Let $P$ and $Q$ in $\F[X]$ of degree at most $d$. $\gcd(P,Q)$ or the gcd-free
part of $P$ with respect to $Q$ can be computed with $\sO(d)$ operations
in~$\F$.
\end{lemma}

\begin{lemma}\label{lem:complexity:shear}
  Let $P$ and $Q$ in $\Z[X,Y]$ be of total degree at most $d$ and maximum bitsize $\tau$. The sheared
  polynomials $P(T-SY,Y)$ and $Q(T-SY,Y)$ can be expanded in $\sOB(d^4+d^3\tau)$ and their bitsizes
  are in $\sO(d+\tau)$. The resultant $R(T,S)$ can be computed in $\sOB(d^7+d^6\tau)$ bit operations
  and $\sO(d^5)$ arithmetic operations in $\Z$; its degree is at most $2d^2$ in each variable and
  its bitsize is in $\sO(d^2+d\tau)$.
\end{lemma}

\begin{proof}
Writing $P$ as $\sum_{i=0}^d{p_i(Y)X^i}$, expending the substitution of $X$ by
  $T-SY$ needs the computation of the successive powers $(T-SY)^i$ for $i$ from
  1 to $d$.  The binomial formula shows that each polynomial $(T-SY)^i$ is the
  sum of $i+1$ monomials, with coefficients of bitsize in $O(i\log{i})$. Using
  the recursion formula $(T-SY)^i=(T-SY)^{i-1}(T-SY)$, given the polynomial
  $(T-SY)^{i-1}$, the computation of $(T-SY)^i$ requires $2i$ multiplications of
  coefficients having bitsize in $O(i\log{i})$, which can be done in
  $\sO_B(i^2\log{i})$ bit operations. The complexity of computing all the powers
  is thus in $\sO_B(d^3\log{d})$.
  The second step is to multiply $p_i(Y)$ by $(T-SY)^i$ for $i=1, \ldots,
  d$. Each polynomial multiplication can be done with $O(d^2)$ multiplications
  of integers of bitsize in $O(\tau)$ or in $O(d\log d)$, and thus it can be
  done in $\sO_B(d^2(\tau+d\log{d}))$ bit operations and yields polynomials of
  bitsize $\OO(\tau+d\log{d})$. For the $d$ multiplications the total cost is in
  $\sO_B(d^3(\tau+d\log{d}))$. Consequently the computation of ${P}(T-SY,Y)$ and
  ${Q}(T-SY,Y)$ can be done in $\sOB(d^3(\tau+d))$ bit operations and these
  polynomials have bitsize in $\sO(\tau+d)$. 
  In addition, since ${P}(T-SY,Y)$ and ${Q}(T-SY,Y)$ are trivariate polynomials of partial degree in
  all variables bounded by $d$, Lemma~\ref{complexity:subresultant} implies the claims on $R(T,S)$.
\end{proof}

\section{Separating linear form}\label{sec:sepelt}

Throughout this section, we assume that the two input polynomials $P$ and $Q$
are coprime in $\Z[X,Y]$, that they define the ideal $I$, that their
  maximum total degree $d$ is at least 2 and that their coefficients have
  maximum bitsize $\tau$. Note that the coprimality of $P$ and $Q$ is implicitly tested during Algorithm
\ref{alg:sep-elem} because they are coprime  if and only if  $R(T,S)$  does  not identically
vanish.
By abuse of notation, some complexity $\sOB(d^k)$ may
  refer to a complexity in which polylogarithmic factors in $d$ \emph{and in
    $\tau$} are omitted.
$I_\mu=\ideal{P_\mu,Q_\mu}$ denotes the ideal generated by $P_\mu=\phi_\mu(P)$
and $Q_\mu=\phi_\mu(Q)$. Similarly as in Equation~\eqref{eq0}, we define
$R_\mu(T,S)$ as the resultant of $P_\mu(T-SY,Y)$ and $Q_\mu(T-SY,Y)$ with
respect to $Y$, and we define $L_{P_\mu}(S)$ and $L_{Q_\mu}(S)$
similarly as in \eqref{eq1}.
We refer to the overview in Section~\ref{sec:overview} for the organization
of this section.

\subsection{Separating linear form over $\Z_\mu$ versus $\Z$}
\label{sec:sep-modulo}

We first introduce the notion of lucky prime numbers $\mu$ which are, roughly
speaking, primes $\mu$ for which the number of distinct solutions of
$\ideal{P,Q}$ does not change when considering the polynomials modulo $\mu$. We
then show the critical property that, 
if a linear form is separating modulo such a $\mu$, then it is also separating over
$\Z$. 

\begin{definition}\label{def:lucky-nb-roots}
  A prime number $\mu$ is said to be {\em \bf lucky} for an ideal $I=\ideal{P,Q}$ if it is
  larger than $2d^4$ and satisfies
\[\phi_\mu(L_P(S))\ \phi_\mu(L_Q(S))\not\equiv 0 \quad\mbox{and}\quad 
\#V(I)=\#V(I_\mu).\]
\end{definition}

\begin{proposition}\label{prop:separation-correspondence}
  Let $\mu$ be a lucky prime for the ideal $I=\ideal{P,Q}$ and let $a<\mu$ be
  an integer\footnote{We assume $a<\mu$ for clarity so that the linear form  $X+aY$ is ``identical''
  in $\Z$ and in $\Z_\mu$. This  hypothesis is however  not needed and we actually prove that if $X+\phi_\mu(a)Y$ separates
$V(I_\mu)$, then $X+aY$ separates $V(I)$.} such that
$\phi_\mu(L_P(a))\ \phi_\mu(L_Q(a))\neq 0.$
If $X+aY$ separates
$V(I_\mu)$, it also separates $V(I)$.
\end{proposition}

The key idea of the proof of Proposition \ref{prop:separation-correspondence},
as well as Propositions \ref{prop:lessroot} and \ref{prop:luckyroot}, is to
prove the following inequalities (under the hypothesis that various leading
terms do not vanish)
\begin{equation}\label{eq:18a}
\#V(I_\mu)\geq d_T(\overline{R_\mu(T,a)})\leq d_T(\overline{R(T,a)})\leq \#V(I)
\end{equation}
and argue that the first (resp. last) one is an
equality if $X+aY$  separates $V(I_\mu)$ (resp. $V(I)$).
We establish these claims in Lemmas~\ref{lem:ineq1} and \ref{lem:ineq2}. 
As mentioned in Section~\ref{sec:overview}, Lemma~\ref{lem:ineq1} is the key property in the classical algorithm for computing a separating
form for $I$, which algorithm we   will  use over $\Z_\mu$ to compute a separating form  for
$I_\mu$ in Section~\ref{sec:gathering}. For completeness, we outline its proof (see \cite[Lemma
16]{det-jsc-2009} or \cite[Proposition 11.23]{BPR06} for details).
Recall that $P$ and $Q$ are assumed to be  coprime but not $P_\mu$ and $Q_\mu$.


\begin{lemma}\label{lem:ineq1}
  If $a\in\Z$ is such that $L_{P}(a)\, L_{Q}(a)\neq 0$ then
  $d_T(\overline{R(T,a)})\leq \#V(I)$ and they are equal if and only if $X+aY$
  separates $V(I)$. The same holds over $\Z_\mu$, that is for $P_\mu$, $Q_\mu$,
  $R_\mu$ and $I_\mu$, provided $P_\mu$ and $Q_\mu$ are coprime.
%
\end{lemma}
\begin{proof}
  Since $L_{P}(a)\ L_{Q}(a)\neq 0$, the resultant $R(T,S)$ can be specialized at
  $S=a$, that is $R(T,a)=Res_Y(P(T-aY,Y), Q(T-aY,Y))$. On the other hand, the
  sheared polynomials $P(T-aY,Y)$ and $Q(T-aY,Y)$ are coprime (since $P$ and $Q$
  are coprime) and since $L_{P}(a)\ L_{Q}(a)\neq 0$, they have no common
  solution at infinity in the $Y$-direction.  Thus the roots of their resultant
  with respect to $Y$ are the $T$-coordinates of the (affine) solutions of
  $I_a=\ideal{P(T-aY,Y), Q(T-aY,Y)}$ 
  (see for instance \cite[\S 3.6 Proposition 3]{CLO}).
  Hence, $d_T(\overline{R(T,a)})\leq \#V(I_a)=\#V(I)$. Moreover, if $X+aY$
  separates $V(I)$, $T=X+aY$ takes distinct values for every solution in $V(I)$,
  and since these values of $T$ are roots of $R(T,a)$,
  $d_T(\overline{R(T,a)})\geq\#V(I)$ and thus they are equal. Conversely, if
  $d_T(\overline{R(T,a)})=\#V(I)$, $R(T,a)$ admits $\#V(I)$ distinct roots
  $T=X+aY$ which means that $X+aY$ separates all the solutions of $V(I)$.  The
  same argument holds over $\Z_\mu$.
\end{proof}


The following lemma states a rather standard properties.  For
  completeness and readers' convenience, we provide a proof
  for which we could not find  accurate references.

%

\begin{lemma}\label{lem:ineq2}
Let $\mu$ be a  prime and   $a$ be an integer  such that 
$\phi_\mu(L_P(a))\ \phi_\mu(L_Q(a))\neq 0$, 
then $d_T(\overline{R_\mu(T,a)})\leq d_T(\overline{R(T,a)})$.
\end{lemma}

\begin{proof}
By hypothesis,
$\phi_\mu(L_P(S))$ and $\phi_\mu(L_Q(S))$ do not
  identically vanish, thus we can specialize the resultant $R$ by $\phi_\mu$, that
  is $\phi_\mu(R(T,S))=Res_Y(\phi_\mu(P(T-SY,Y)), \phi_\mu(Q(T-SY,Y)))$ \cite[Proposition 4.20]{BPR06}. Hence,
$\phi_\mu(R(T,S))=R_\mu(T,S)$. The evaluation at $S=a$ and the reduction modulo $\mu$
  commute (in $\Z_\mu$), thus  $\phi_\mu(R(T,a))=R_\mu(T,a)$  in $\Z_\mu[T]$. 

We  now show that for any polynomial $f \in \Z[X]$ and prime $\mu$, $\deg(\overline{\phi_\mu(f)})
\leq \deg(\overline{f})$, which will implie the lemma. 

Let $f=c\prod_{i}{f_i^{m_i}}$ be the squarefree decomposition of $f$ in $ \Z[X]$. Considering its
reduction modulo $\mu$, we obtain that $\phi_\mu(f)=\phi_\mu(c)\prod_{i}{\phi_\mu(f_i)^{m_i}}$.
Hence, $\deg(\overline{\phi_\mu(f)}) \leq \sum_i{\deg(\phi_\mu(f_i))}$. Furthermore, since $\deg(\phi_\mu(f_i)) \leq \deg(f_i)$, we have that $\deg(\overline{\phi_\mu(f)}) \leq
\sum_i{\deg(f_i)}$.
On the other hand, since $f=c\prod_{i}{f_i^{m_i}}$ is  the squarefree decomposition of $f$, we have
$\deg(\overline{f})=\sum_i{\deg(f_i)}$ so $\deg(\overline{\phi_\mu(f)}) \leq \deg(\overline{f})$.
\end{proof}

\begin{proof}[Proof of Proposition~\ref{prop:separation-correspondence}]
If $\mu$ is a lucky prime, then by definition $\#V(I)=\#V(I_\mu)$, thus $I_\mu$ is zero-dimensional since $I$
is. Thus, by Lemmas~\ref{lem:ineq1} and \ref{lem:ineq2}, if $\mu$ is a lucky prime and $a$ is an integer such that $X+aY$ separates
$V(I_\mu)$ and $\phi_\mu(L_P(a))\ \phi_\mu(L_Q(a))\neq 0$, then
\[\#V(I_\mu)= d_T(\overline{R_\mu(T,a)})\leq d_T(\overline{R(T,a)})\leq \#V(I).\]
Since $\mu$ is lucky, $\#V(I_\mu)= \#V(I)$ thus $d_T(\overline{R(T,a)})= \#V(I)$ and by Lemma~\ref{lem:ineq1}, $X+aY$ separates~$V(I)$.
\end{proof}

\subsection{Number of solutions  over $\Z_\mu$ versus $\Z$}\label{sec:lucky-nbroot}

As shown in Proposition \ref{prop:separation-correspondence}, the knowledge of a
lucky prime permits to search for separating linear forms over $\Z_\mu$ rather
than over $\Z$. We prove here two propositions that are critical for computing a
lucky prime, which state that the number of solutions of
$I_\mu=\ideal{P_\mu,Q_\mu}$ is always at most that of
$I=\ideal{P,Q}$ 
and give a bound on the number of unlucky primes. 

\begin{proposition}\label{prop:lessroot}
  Let $I=\ideal{P,Q}$ be a zero-dimensional ideal in $\Z[X,Y]$. If a prime $\mu$
  is larger than $2d^4$ such that $I_\mu$ is zero-dimensional and
$\phi_\mu(L_P(S))\ \phi_\mu(L_Q(S)) \not\equiv 0$
then $\#V(I_\mu)\leq \#V(I)$.  
\end{proposition}
\begin{proof}
  Let $\mu$ be a prime that satisfies the hypotheses of the proposition.  We
  also consider an integer $a<\mu$ such that 
$\phi_\mu(L_P(a))\  \phi_\mu(L_Q(a))\neq 0$ 
and such that the linear form $X+aY$
  is separating for $I_\mu$. Such an integer exists because  (i)
  $\phi_\mu(L_P(S))$ and $\phi_\mu(L_Q(S))$ 
are not
  identically zero by hypothesis and they have degree at most $d$ 
and, 
 since $I_\mu$ is zero dimensional, (ii) $I_\mu$  
has at most $d^2$ solutions which define
  at most $d^2\choose 2$ directions in which two solutions are aligned.
  Since $2d+{d^2\choose 2}< 2d^4$ (for $d\geq 2$), there exists such an
  integer $a\leq 2d^4<\mu$.  With such an $a$, we can apply
  Lemmas~\ref{lem:ineq1} and \ref{lem:ineq2} which imply that
  $\#V(I_\mu)=d_T(\overline{R_\mu(T,a)})\leq d_T(\overline{R(T,a)})\leq \#V(I)$.
\end{proof}
%

Next, we bound the number of primes that are unlucky for the ideal $\ideal{P,Q}$.

\begin{proposition}\label{prop:luckyroot}
  An upper bound on the number of unlucky primes for the ideal $\ideal{P,Q}$
  can be explicitly computed in terms of $d$ and~$\tau$, and this bound is in
  $\sO(d^4+d^3\tau)$.
\end{proposition}

\begin{proof}
  According to Definition~\ref{def:lucky-nb-roots}, a prime $\mu$ is unlucky if
  it is smaller than $2d^4$, if $\phi_{\mu}({L_P(S)})$ $\phi_{\mu}(L_Q(S))\not\equiv 0$, or if
  $\#V(I)\neq \#V(I_\mu)$. In the following, we consider $\mu>2d^4$.  We first
  determine some conditions on $\mu$ that ensure that $\#V(I)= \#V(I_\mu)$, and
  we then bound the number of $\mu$ that do not satisfy these conditions.  As we
  will see, under these conditions, $L_P(S)$ and $L_Q(S)$ do not
  vanish modulo $\mu$ and thus this constraint is redundant.

The first part of the proof is similar in spirit to that of
Proposition~\ref{prop:lessroot} in which we first fixed a prime $\mu$ and then
specialized the polynomials at $S=a$ such that the form $X+aY$ was separating
for $I_\mu$. Here, we first choose $a$ such that $X+aY$ is separating for
$I$. With some conditions on $\mu$, Lemmas~\ref{lem:ineq1} and \ref{lem:ineq2}
imply Equation~\eqref{eq:21a} and we determine some more conditions on $\mu$
such that the middle inequality of \eqref{eq:21a} is an equality. We thus get
$\#V(I_\mu)\geq \#V(I)$ which is the converse of that of
Proposition~\ref{prop:lessroot} and thus $\#V(I_\mu)=\#V(I)$. In the second part
of the proof, we bound the number of $\mu$ that violate the conditions we
considered.

\medskip\noindent\emph{Prime numbers such that ${\#V(I)\neq \#V(I_\mu)}$.}\quad Let $a$ be such that
the form $X+aY$ separates $V(I)$ and $L_{P}(a)\,L_{Q}(a)\,L_{R}(a)\neq 0$.\footnote{It can be shown
  that $L_{P}(a)\,L_{Q}(a)\neq 0$ implies $L_{R}(a)\neq 0$ (see for instance \cite[Lemma
  11]{bouzidi2013RurWorstCase}) but this property does not simplify the proof.} 
 Similarly as in the
proof of Proposition~\ref{prop:lessroot}, since $L_R(S)$ has degree at most $2d^2$
(Lemma~\ref{complexity:subresultant}) and $2d+2d^2+{d^2\choose 2}< 2d^4$ (for $d\geq 2$), we can choose $a\leq 2d^4$.
 
\emph{We consider any prime $\mu>2d^4$ such that 
$\phi_\mu(L_P(a))\ \phi_\mu(L_Q(a))\ \phi_\mu(L_R(a))\neq 0$}. By
Lemmas~\ref{lem:ineq1} and \ref{lem:ineq2}, we have
\begin{equation}\label{eq:21a}
\#V(I_\mu)\geq d_T(\overline{R_\mu(T,a)})\leq d_T(\overline{R(T,a)})= \#V(I),
\end{equation}
since the first inequality trivially holds when $I_\mu$ is not zero-dimensional and since   $X+aY$ separates $V(I)$.


Now, $d_T(\overline{R(T,a)})= d_T(R(T,a)) - d_T(\gcd(R(T,a),R'(T,a)))$, and
similarly for $R_\mu(T,a)$.  The leading coefficient of $R(T,S)$ with respect to
$T$ is $L_R(S)$, and since it does not vanish at $S=a$, $L_R(a)$ is the leading
coefficient of $R(T,a)$. 
In addition, since $\phi_\mu(L_P(a))\,
\phi_\mu(L_Q(a))\neq 0$, we can specialize the resultant $R$ by $\phi_\mu$, thus
   $\phi_\mu(R(T,a))=Res_Y(\phi_\mu(P(T-aY,Y)), \phi_\mu(Q(T-aY,Y)))$ \cite[Proposition 4.20]{BPR06}. Hence,
$\phi_\mu(R(T,a))=R_\mu(T,a)$ and 
the hypothesis $\phi_\mu(L_R(a))\neq 0$ implies that $R_\mu(T,a)$ and $R(T,a)$
have the same degree.
It follows that, \emph{if $\mu$ is such that the degree of
  $\gcd(R(T,a),R'(T,a))$ does not change when $R(T,a)$ and $R'(T,a)$ are reduced
modulo $\mu$}, we have
\[ \#V(I_\mu)\geq d_T(\overline{R_\mu(T,a)}) = d_T(\overline{R(T,a)})= \#V(I).\]

Since $\phi_\mu(R(T,a))=R_\mu(T,a)$ and  $\phi_\mu(L_R(a)) \neq 0$, the resultant $R_\mu(T,a)$ does
not identically vanish and thus $I_\mu$ is zero-dimensional. Furthermore, 
since $\mu>2d^4$ and $\phi_\mu(L_P(a))\ \phi_\mu(L_Q(a)) \neq 0$, we can apply
Proposition~\ref{prop:lessroot} which yields that $\#V(I_\mu)\leq \#V(I)$ and thus  $\#V(I_\mu)=\#V(I)$.


Therefore, the primes $\mu$ such that $\#V(I_\mu)\neq\#V(I)$ are among those
such that $\mu\leq 2d^4$, or  $L_P(a)$, $L_Q(a)$ or $L_R(a)$ vanishes modulo $\mu$ or such
that the degree of $\gcd(R(T,a),R'(T,a))$ changes when $R(T,a)$ and
$R'(T,a)$ are reduced modulo $\mu$.  Note that if $L_P(a)$ and $L_Q(a)$ do not vanish modulo $\mu$,
then $L_P(S)$ and $L_Q(S)$ do
not identically vanish modulo $\mu$.

\medskip\noindent\emph{Bounding the number of prime divisors of  $L_P(a)$, $L_Q(a)$ or $L_R(a)$.}\quad
The number of prime divisors of an integer $z$ is bounded by its bitsize. Indeed, its bitsize is $\lfloor \log
z \rfloor +1$ and its factorization
into $w$ (possibly identical) prime numbers directly yields that $2^w\leq \prod_{i=1}^{w}z_i =z=
2^{\log z}\leq 2^{\lfloor \log z \rfloor +1}$. We can thus bound the number of prime divisors by
bounding the bitsize of 
$L_P(a)$, $L_Q(a)$ and $L_R(a)$.  We start by bounding the bitsize of 
$L_P(S)$, $L_Q(S)$ and $L_R(S)$.

 Each coefficient of $P(T-SY,Y)$ has bitsize at most $\tau'=\tau+d\log d +\log(d+1)+1$. Indeed,  $(T-SY)^i$
 is a sum of $i+1$ monomials whose coefficients are  binomials $\binom{i\leq d}{j}<d^d$.
The claim follows since each coefficient of $P(T-SY,Y)$ is the sum of at most $d+1$ such binomials, each
 multiplied by a coefficient of $P(X,Y)$ which has bitsize at most $\tau$.
We get the same bound for the coefficients of $Q(T-SY,Y)$ and thus for $L_P(S)$ and $L_Q(S)$ as
well. Concerning $L_R(S)$, we have that $R(T,S)$ is the resultant
  of $P(T-SY,Y)$ and $Q(T-SY,Y)$ thus, by Lemma~\ref{complexity:subresultant},
  its coefficients are of bitsize $\sO(d\tau')$. 
In fact, an upper bound
  can be explicitly computed using, for
  instance, the bound of \cite[Theorem 8.46]{BPR06} which implies that the resultant of two trivariate
  polynomials of total degree $d'$ and bitsize $\tau'$ has bitsize at most $2d'(\tau'+\lfloor\log
  2d'\rfloor +1)+2(\lfloor\log (2d'^2+1)\rfloor +1)$, which is in $\sO(d^2+d\tau)$ in our case. 
Therefore, $L_P(S)$, $L_Q(S)$ and $L_R(S)$ have degree at most $2d^2$ and their bitsizes can be
explicitly  bounded by a function of $d$ and $\tau$ in $\sO(d^2+d\tau)$.

Finally, since $a\leq 2d^4$, its bitsize is at most $\sigma=4\log d +2$. It is
straightforward that the result of an evaluation of a univariate polynomial of degree at most $d'$ and
bitsize $\tau'$ at an integer value of bitsize $\sigma$ has bitsize at most
$d'\sigma+\tau'+\log(d'+1)+1$. Here $d'\leq 2d^2$ and $\tau'$ is in $\sO(d^2+d\tau)$.  We thus
proved that we can compute an explicit bound, in $\sO(d^2+d\tau)$, on the number of prime divisors
of $L_P(a)$, $L_Q(a)$, or $L_R(a)$.

\medskip\noindent\emph{Bounding the number of prime $\mu$ such that the degree
  of $\gcd(R(T,a),R'(T,a))$ changes when $R(T,a)$ and $R'(T,a)$ are reduced
  modulo $\mu$. }\quad By \cite[Lemma 4.12]{Yap-2000}, given two univariate
polynomials in $\Z[X]$ of degree at most $d'$ and bitsize at most $\tau'$, the
degree of their gcd changes when the polynomials are considered modulo $\mu$ on
a set of $\mu$ whose product is bounded\footnote{\cite[Lemma 4.12]{Yap-2000}
  states the bound as $N^{2d'+2}$ where $N$ is the maximum Euclidean norm of the
  vectors of coefficients of the polynomials.} by
$(2^{\tau'}\sqrt{d'+1})^{2d'+2}$. As noted above, the number of such primes
$\mu$ is bounded by the bitsize of this bound, and thus is bounded by
$(d'+1)\,(2\tau'+\log(d'+1))+1$.  Here $d'\leq 2d^2$ and $\tau'$ is in
$\sO(d^2+d\tau)$ since our explicit bound on the bitsize of $L_R(a)$ holds as
well for the bitsize of $R(T,a)$, and, since $R(T,a)$ is of degree at most
$2d^2$, the bitsize of $R'(T,a)$ is bounded by that of $R(T,a)$ plus $1+\log
2d^2$. We thus obtain an explicit bound in $\sO(d^4+d^3\tau)$ on the number of
primes $\mu$ such that the degree of $\gcd(R(T,a),R'(T,a))$ changes when
$R(T,a)$ and $R'(T,a)$ are reduced modulo $\mu$.

The result follows by summing  this bound with the bounds we obtained  on the number of
prime divisors of $L_P(a)$, $L_Q(a)$, or $L_R(a)$, and a bound (e.g. $2d^4$) on the number of primes smaller than~$2d^4$.
\end{proof}




\subsection{Counting the number of solutions over $\Z_\mu$}
\label{sec:counting-Zmu}

For counting the number of (distinct) solutions of
$\ideal{P_\mu,Q_\mu}$, we  use a classical algorithm for computing a triangular decomposition of
an ideal defined by two bivariate polynomials. We first  recall  this algorithm, slightly
adapted to our needs, and analyze its arithmetic complexity.

\subsubsection{Triangular decomposition}
\label{sec:tri-dec}

Let $P$ and $Q$ be two polynomials in $\F[X,Y]$.
A decomposition of the solutions of the system $\{P,Q\}$ using the
subresultant sequence appears in the theory of triangular sets
\cite{lazard1991,Li-modpn-11} and for the computation of topology of curves
\cite{VegKah:curve2d:96}. 


The idea is to use Lemma~\ref{lem:fund-prop-subres} which states that, after
specialization at $X=\alpha$, the first (with respect to increasing $i$) nonzero
subresultant $Sres_{Y,i}(P,Q)(\alpha,Y)$ is of degree $i$ and is equal to the
gcd of $P(\alpha,Y)$ and $Q(\alpha,Y)$.  This induces a decomposition of the
system $\{P,Q\}$ into triangular subsystems $(\{A_i(X),Sres_{Y,i}(P,Q)(X,Y)\})$
 where a solution $\alpha$ of $A_i(X)=0$ is such that 
the system $\{P(\alpha,Y), Q(\alpha,Y)\}$ admits exactly $i$ roots (counted with multiplicity), 
which  are exactly those of $Sres_{Y,i}(P,Q)(\alpha,Y)$.
Furthermore, these triangular subsystems are regular chains, i.e.,  the
leading coefficient of the bivariate polynomial (seen in  $Y$)  is coprime with the univariate
polynomial. 
For clarity and self-containedness, we recall this decomposition in Algorithm~\ref{alg:tri-dec-mod},
where, in addition, we restrict the solutions of the system $\{P, Q\}$ to those where some univariate polynomials $A(X)$
vanishes ($A$ could be identically zero).

\begin{algorithm}[t]
  \caption{
    Triangular decomposition \cite{VegKah:curve2d:96,Li-modpn-11}}
\label{alg:tri-dec-mod}
\begin{algorithmic}[1]
  \REQUIRE{ $P,Q$ in $\F[X,Y]$ coprime such that $Lc_Y(P)$ and $Lc_Y(Q)$ are coprime,\footnotemark $d_Y(Q)\leq d_Y(P)$, and \\
\qquad $A \in \F[X]$ squarefree.}
  \ENSURE{Triangular decomposition 
  $\{(A_i(X),$ $B_i(X,Y))\}_{i\in\cal I}$ such that $V(\ideal{P,Q,A})$ is the disjoint union
  of the sets $V(\ideal{A_i(X),B_i(X,Y)})_{i\in\cal I}$}
%
%
\STATE Compute the subresultant sequence of $P$ and $Q$ with respect to $Y$: $B_i=Sres_{Y,i}(P,Q)$
\STATE  $G_0=\gcd(\overline{Res_Y(P,Q)},A)$ and ${\calT}=\emptyset$
\FOR {$i=1$ \TO $d_Y(Q)$}
\STATE $G_i=\gcd(G_{i-1},sres_{Y,i}(P,Q))$

\STATE $A_i=G_{i-1}/G_i$
\STATE if $d_X(A_i)>0$, add $(A_i,B_i)$ to ${\calT}$ 
\ENDFOR
\RETURN ${\calT}= \{(A_i(X),B_i(X,Y))\}_{i\in{\cal I}}$

\end{algorithmic}
\end{algorithm}
\footnotetext{The hypothesis that $Lc_Y(P)$ and $Lc_Y(Q)$ are coprime can be
  relaxed by applying the algorithm recursively (see \cite{Li-modpn-11} for
  details). We require here this hypothesis for complexity issues.}

The following lemma states the correctness of Algorithm~\ref{alg:tri-dec-mod}
which follows from Lemma~\ref{lem:fund-prop-subres} and from the fact that the
solutions of $P$ and $Q$ project on the roots of their resultant.
%

\begin{lemma}[{\cite{VegKah:curve2d:96,Li-modpn-11}}]\label{lem:tridec-correctness}
  Algorithm~\ref{alg:tri-dec-mod} computes a triangular decomposition 
  $\{(A_i(X),$ $B_i(X,Y))\}_{i\in\cal I}$ such that 
\begin{itemize}\cramped
  \item[(i)] the set $V(\ideal{P,Q,A})$ is the disjoint union
  of the sets $V(\ideal{A_i(X),B_i(X,Y)})_{i\in\cal I}$,
  \item[(ii)] $\prod_{i\in\cal I}A_i$ is
  squarefree, 
  \item[(iii)] $\forall \alpha\in V(A_i)$,
      $B_i(\alpha,Y)$ is of degree $i$ and is equal to $\gcd(P(\alpha,Y),$
      $Q(\alpha,Y))$, and 
  \item[(iv)] $A_i(X)$ and $Lc_Y(B_i(X,Y))$ are coprime.
\end{itemize}
\end{lemma}

In the following lemma, we analyze the complexity of
Algorithm~\ref{alg:tri-dec-mod} for $P$ and $Q$ of degree at most $d_X$ in $X$
and $d_Y$ in $Y$ and $A$ of degree at most $d^2$, where $d$ denotes a bound on
the total degree of $P$ and $Q$.  We will use Algorithm~\ref{alg:tri-dec-mod}
with polynomials with coefficients in $\F=\Z_\mu$ and we thus only consider its
arithmetic complexity in $\F$.
Note that the bit complexity of this algorithm, over $\Z$,
is analyzed in \cite[Theorem 19]{det-jsc-2009} and its arithmetic complexity is
thus implicitly analyzed as well; for clarity, we provide here a short proof.

\begin{lemma}\label{lem:tridec}
  Algorithm~\ref{alg:tri-dec-mod} performs $\sO(d_Xd_Y^3)=\sO(d^4)$ arithmetic
  operations in $\F$. 
\end{lemma}
\begin{proof}
  From Lemma \ref{complexity:subresultant} (note that this lemma is stated for
  the coefficient ring $\Z$, but the arithmetic complexity is the same for any
  field $\F$), the subresultant sequence of $P$ and $Q$ can be computed in
  $\sO(d_Xd_Y^3)$ arithmetic operations, and the resultant as well as the
  principal subresultant coefficients have degrees in $\OO(d_Xd_Y)$.  The
  algorithm performs at most $d_Y$ gcd computations between these univariate
  polynomials. The arithmetic complexity of one such gcd computation is soft
  linear in their degrees, that is $\sO(d_Xd_Y)$
  (Lemma~\ref{complexity:gcd}). Hence the arithmetic complexity of computing the
  systems $\{S_i\}_{i=1 \ldots d}$ is $\sO(d_Xd_Y^2)$.  The total complexity of
  the triangular decomposition is hence dominated by the cost of the subresultant
  computation, that is $\sO(d_Xd_Y^3)=\sO(d^4)$.
\end{proof}

\subsubsection{Counting the number of solutions over $\Z_\mu$}

Algorithm \ref{alg:rad-tri-dec} computes the number of distinct solutions of an
ideal $I_\mu=\ideal{P_\mu,Q_\mu}$ of $\Z_\mu[X,Y]$.  Roughly speaking, this algorithm first performs 
one triangular decomposition with the input polynomials $P_\mu$
and $Q_\mu$, and then performs a sequence of triangular decompositions with polynomials resulting from this
decomposition.
The result is close to a radical triangular decomposition and the
number of solutions of $I_\mu$ can be read, with a simple formula, from the
degrees of the polynomials in the decomposition.  Note that
Algorithm~\ref{alg:rad-tri-dec}, as Algorithm \ref{alg:tri-dec-mod}, is 
valid for any base field $\F$ but,  since we will only use it over $\Z_\mu$, we
state it and analyze its complexity in this case.

\begin{algorithm}[t]
  \caption{Number of distinct solutions of  $\ideal{P_\mu,Q_\mu}$}
\label{alg:rad-tri-dec}
\begin{algorithmic}[1]
  \REQUIRE{$P_\mu,Q_\mu$ in $\Z_\mu[X,Y]$ coprime, $\mu$ larger than their total degree}
  \ENSURE{ Number of distinct solutions of $\ideal{P_\mu,Q_\mu}$}
\medskip
\STATE Shear $P_\mu$ and $Q_\mu$ by replacing $X$ by $X-bY$ with $b\in \Z_\mu$  so that
$Lc_Y(P_\mu(X-bY,Y)) \in\Z_\mu$  \label{alg:nbroot-shear}

\STATE Triangular decomposition: 
$\{(A_i(X),B_i(X,Y))\}_{i\in\cal I}=$ Algorithm
\ref{alg:tri-dec-mod} $(P_\mu,Q_\mu,0)$  \label{alg:nbroot-1stdecomp}

\FORALL {$i\in {\cal I}$ \label{alg:nbroot-loop}} 
   \STATE $C_i(X)= Lc_Y(B_i(X,Y))^{-1} \bmod A_i(X)$ \label{alg:nbroot-Ci}
   \STATE $\tilde{B}_i(X,Y)=C_i(X) B_i(X,Y) \bmod A_i(X)$ \label{alg:nbroot-Bimod}
   \STATE    Triangular decomp.: \\$\{ (A_{ij}(X),B_{ij}(X,Y))\}_{{j\in\cal J}_i}=$
Algorithm~\ref{alg:tri-dec-mod} $\left(\tilde{B}_i(X,Y),\frac{\partial \tilde{B}_i(X,Y)}{\partial
  Y}, A_i(X)\right)$\label{alg:nbroot-2nddecomp}
\ENDFOR
\RETURN $\sum_{i\in\cal I}\left(i\,d_X(A_i)-\sum_{j\in{\cal J}_i} j\,d_X(A_{ij})\right)$\label{alg:nbroot-return}
\end{algorithmic}
\end{algorithm}

\begin{lemma}\label{lem:proof-terminaison}
  Algorithm \ref{alg:rad-tri-dec} computes the number of distinct solutions of $\ideal{P_\mu,Q_\mu}$.
\end{lemma}

\begin{proof} 
  The shear of Line \ref{alg:nbroot-shear} allows to fulfill the requirement of the triangular
  decomposition algorithm, called in Line~\ref{alg:nbroot-1stdecomp}, that the input polynomials have
  coprime leading coefficients. 
  Once the generically sheared polynomial $P_\mu(X-SY,Y)$ is computed (in
  $\Z_\mu[S,X,Y]$), a specific shear value $b\in \Z_\mu$ can be selected by
  evaluating the univariate polynomial $L_{P_\mu}(S)=Lc_Y(P_\mu(X-SY,Y))$ at
  $d+1$ elements of $\Z_\mu$. The polynomial does not vanish at one of these
  values since it is of degree at most $d$ and $d<\mu$. Note that such a shear clearly does not change the number of solutions.

  According to Lemma~\ref{lem:tridec-correctness}, the triangular decomposition
  $\{(A_i(X),B_i(X,Y))\}_{i\in\cal I}$ computed in
  Line~\ref{alg:nbroot-1stdecomp} is such that the solutions of
  $\ideal{P_\mu,Q_\mu}$ is the disjoint union of the solutions of the
  $\ideal{A_i(X),B_i(X,Y)}$, for $i\in\cal I$.  It follows that the number of
  (distinct) solutions of $I_\mu=\ideal{P_\mu,Q_\mu}$ is
\[\#V(I_\mu) = \sum_{i\in\cal I} \sum_{\alpha\in V(A_i)}
d_Y(\overline{B_i(\alpha,Y)}).\]

Since $B_i(\alpha,Y)$ is a univariate polynomial in $Y$,\\
$d_Y(\overline{B_i(\alpha,Y)}) = d_Y({B_i(\alpha,Y)}) -
d_Y(\gcd(B_i(\alpha,Y),B'_i(\alpha,Y)))$, where $B'_i(\alpha,Y)$ is the
derivative of $B_i(\alpha,Y)$, which is also equal to $\frac{\partial
  B_i}{\partial Y}(\alpha,Y)$.  By Lemma~\ref{lem:tridec-correctness},
$d_Y(B_i(\alpha,Y))=i$, and since the degree of the gcd is zero when
$B_i(\alpha,Y)$ is squarefree, we have
\begin{equation}\label{eq:countmu1}
\#V(I_\mu) = \sum_{i\in\cal I} \left( \sum_{\alpha\in V(A_i)} i - 
\sum_{\substack{\alpha\in
V(A_i)\\B_i(\alpha,Y) \mbox{\footnotesize ~not sqfr.}}}  
d_Y(\gcd(B_i(\alpha,Y),\mbox{$\frac{\partial B_i}{\partial Y}$}(\alpha,Y)))\right).
\end{equation}
The polynomials $A_i(X)$ are squarefree by Lemma~\ref{lem:tridec-correctness}, so  $\sum_{\alpha\in V(A_i)} i$ is  equal to $i\, d_X(A_i)$.

We now consider the sum of the degrees of the gcds. The  rough idea is to apply Algorithm~\ref{alg:tri-dec-mod} to $B_i(X,Y)$ and $\frac{\partial B_i}{\partial
  Y}(X,Y)$, for every
$i\in\cal I$, which computes a triangular decomposition $\{(A_{ij}(X),B_{ij}(X,Y))\}_{j\in{\cal
J}_i}$ such that, for $\alpha\in V(A_{ij})$,  $d_Y(\gcd(B_i(\alpha,Y),\frac{\partial B_i}{\partial
  Y}(\alpha,Y)))=j$ (by Lemma~\ref{lem:tridec-correctness}), which simplifies Equation~\eqref{eq:countmu1} into $\#V(I_\mu) = \sum_{i\in\cal I} \left( i\, d_X(A_i) - \sum_{j\in{\cal J}_i}\ \sum_{\alpha\in V(A_{ij})}
j\right)$. However, we cannot directly apply Algorithm~\ref{alg:tri-dec-mod} to $B_i(X,Y)$ and $\frac{\partial B_i}{\partial
  Y}(X,Y)$ because their leading coefficients in $Y$  have no
reason to be coprime.

By Lemma~\ref{lem:tridec-correctness}, $A_i(X)$ and $Lc_Y(B_i(X,Y))$ are
coprime, thus $Lc_Y(B_i(X,Y))$ is invertible modulo $A_i(X)$ (by Bézout's
identity); let $C_i(X)$ be this inverse and define $\tilde{B}_i(X,Y)$ $=C_i(X)
B_i(X,Y)$ $ \bmod A_i(X)$ (such that every coefficient of $C_i(X) B_i(X,Y)$ with
respect to $Y$ is reduced modulo $A_i(X)$).
The leading coefficient in $Y$ of $\tilde{B}_i(X,Y)$ is equal to 1, so we can
apply Algorithm~\ref{alg:tri-dec-mod} to $\tilde{B}_i(X,Y)$ and $\frac{\partial
  \tilde{B}_i}{\partial Y}(X,Y)$.  Furthermore, if $A_i(\alpha)=0$, then
$\tilde{B}_i(\alpha,Y)=C_i(\alpha) B_i(\alpha,Y)$ where $C_i(\alpha)\neq 0$
since $C_i(\alpha)Lc_Y(B_i(\alpha,Y))=1$. Equation~\eqref{eq:countmu1} can thus
be rewritten by replacing $B_i$ by $\tilde{B}_i$.

By Lemma~\ref{lem:tridec-correctness}, for every $i\in\cal I$,
Algorithm~\ref{alg:tri-dec-mod} computes a triangular decomposition
$\{(A_{ij}(X),$ $B_{ij}(X,Y))\}_{j\in{\cal J}_i}$ such that
$V(\ideal{\tilde{B}_i,\frac{\partial \tilde{B}_i}{\partial Y},A_i})$ is the
disjoint union of the sets $V(\ideal{A_{ij}(X),B_{ij}(X,Y)})$, $j\in{\cal J}_i$,
and for all $ \alpha\in V(A_{ij})$,
$d_Y(\gcd(\tilde{B}_i(\alpha,Y),\frac{\partial \tilde{B}_i}{\partial
  Y}(\alpha,Y)))=j$.  Since the set of $\alpha\in V(A_i)$ such that
$\tilde{B}_i(\alpha,Y)$ is not squarefree is the projection of the set of
solutions $(\alpha,\beta)\in V(\ideal{\tilde{B}_i,\frac{\partial
    \tilde{B}_i}{\partial Y},A_i})$ we get
\[\#V(I_\mu) = \sum_{i\in\cal I} \left( i\, d_X(A_i) - \sum_{j\in{\cal J}_i}\ \sum_{\alpha\in V(A_{ij})}
j\right).\]
$A_{ij}(X)$ is squarefree (Lemma~\ref{lem:tridec-correctness}) so $\sum_{\alpha\in V(A_{ij})}
j=j\, d_X(A_{ij})$, which concludes the proof.
\end{proof}

The next lemma gives the arithmetic complexity of the above algorithm. 

\begin{lemma}\label{lem:cost-arith}
  Given $P_\mu,Q_\mu$ in $\Z_\mu[X,Y]$ of total degree at most $d$, Algorithm
  \ref{alg:rad-tri-dec} performs
$\sO(d^4)$ operations in $\Z_\mu$.
\end{lemma}

\begin{proof}
  According to Lemma~\ref{lem:complexity:shear}, the sheared polynomials
  $P(T-SY,Y)$ and $Q(T-SY,Y)$ can be expanded in $\sOB(d^4+d^3\tau)$ bit
  operations in $\Z$. Thus the sheared polynomials $P_\mu(X-SY,Y)$ and
  $Q_\mu(X-SY,Y)$ can obviously be computed in $\sO(d^4)$ arithmetic operations
  in $\Z_\mu$.\footnote{It can easily be proved that these polynomials can be
    computed in $\sO(d^3)$ arithmetic operations but the $\sO(d^4)$ bound is
    sufficient here.}  The leading term $Lc_Y(P_\mu(X-SY,Y))\in\Z_\mu[S]$ is a
  polynomial of degree at most $d$ and a value $b\in \Z_\mu$ that does not
  vanish it can be found by at most $d+1$ evaluations. Each evaluation can be
  done with $O(d)$ arithmetic operations, thus the shear value $b$ can be
  computed in $\sO(d^2)$ operations. It remains to evaluate the generically
  sheared polynomials at this value $S=b$. These polynomials have $O(d^2)$
  monomials in $X$ and $Y$, each with a coefficient in $\Z_\mu[S]$ of degree at
  most $d$; since the evaluation of each coefficient is soft linear in $d$, this
  gives a total complexity in $\sO(d^4)$ for Line~\ref{alg:nbroot-shear}.

  According to Lemma \ref{lem:tridec}, the triangular decomposition in Line
  \ref{alg:nbroot-1stdecomp} can be done in $\sO(d^4)$ arithmetic operations.
  In Lines \ref{alg:nbroot-Ci} and \ref{alg:nbroot-Bimod}, $C_i(X)$ and
  $\tilde{B}_i(X,Y)$ can be computed by first reducing modulo $A_i(X)$ every
  coefficient of ${B}_i(X,Y)$ (with respect to $Y$).  There are at most $i$
  coefficients (by definition of subresultants) and the arithmetic complexity of
  every reduction is soft linear in the degree of the operands \cite[Corollary
  11.6]{vzGGer}, 
  which is $\sO(d^2)$ by Lemma~\ref{complexity:subresultant}. The reduction of
  ${B}_i(X,Y)$ modulo $A_i(X)$ can thus be done with $\sO(d^3)$ arithmetic
  operations in $\Z_\mu$. Now, in Line \ref{alg:nbroot-Ci}, the arithmetic
  complexity of computing the inverse of one of these coefficients modulo
  $A_i(X)$ is soft linear in its degree \cite[Corollary 11.8]{vzGGer}, that is
  $\sO(d_i)$ where $d_i$ denotes the degree of $A_i(X)$. Furthermore, computing
  the product modulo $A_i(X)$ of two polynomials which are already reduced
  modulo $A_i(X)$ can be done in $\sO(d_i)$ arithmetic operations
  \cite[Corollary 11.8]{vzGGer}. Thus, in Line \ref{alg:nbroot-Bimod}, the
  computation of $\tilde{B}_i(X,Y)$ can be done with $i$ such multiplications,
  and thus with $\sO(id_i)$ arithmetic operations. Finally, in Line
  \ref{alg:nbroot-2nddecomp}, the triangular decomposition can be done with
  $\sO(i ^3d_i)$ arithmetic operations by Lemma~\ref{lem:tridec}. The complexity
  of Lines \ref{alg:nbroot-Ci}-\ref{alg:nbroot-2nddecomp} is thus in $\sO(d^3+i
  ^3d_i)$ which is in $\sO(d^3+d^2i d_i)$. The total complexity of the loop in
  Line \ref{alg:nbroot-loop} is thus $\sO(d^4+d^2\sum_i id_i)$ which is in
  $\sO(d^4)$ because the number of solutions of the triangular system
  $(A_i(X),B_i(X,Y))$ is at most the degree of $A_i$ times the degree of $B_i$
  in $Y$, that is $id_i$, and the total number of these solutions for $i\in \cal
  I$ is that of $(P,Q)$, by Lemma~\ref{lem:tridec-correctness}, which is at most
  $d^2$ by Bézout's bound.  This concludes the proof because the sum in Line
  \ref{alg:nbroot-return} can obviously be done in linear time in the size of
  the triangular decompositions that are computed during the algorithm.
\end{proof}

\subsection{Computing a lucky prime and the number of solutions over $\Z$ } 
\label{sec:gather-lucky}

We now show how to compute the number of solutions of $I=\ideal{P,Q}$ over $\Z$
and a lucky prime for that ideal.

\begin{algorithm}[t]
  \caption{Number of distinct solutions and lucky prime for  $\ideal{P,Q}$} 
\label{algo:nbroot}
\begin{algorithmic}[1]
  \REQUIRE{ $P,Q$ in $\Z[X,Y]$ coprime of total degree at most $d$ and bitsize at most $\tau$}

\ENSURE{The number of solutions and a lucky prime $\mu$ for $\ideal{P,Q}$}
\medskip

\STATE Compute $P(T-SY,Y)$ and $ Q(T-SY,Y)$\label{alg:nbrootZ1}

\STATE Compute a set $B$ of primes larger than $2d^4$ and of
cardinality $\sO(d^4+d^3\tau)$ that contains a lucky prime for
$\ideal{P,Q}$ (see Proposition \ref{prop:luckyroot})\label{alg:nbrootZ2}


\FORALL {$\mu$ in $B$\label{alg:nbrootZ3}}

\STATE Compute the reduction modulo $\mu$ of $P, Q, L_P(S), L_Q(S)$ and
$Res_Y(\phi_\mu(P),\phi_\mu(Q))$ \label{alg:nbrootZ3.1}

\IF{  $Res_Y(\phi_\mu(P),\phi_\mu(Q))\not\equiv0$ and $\phi_\mu(L_P(S))\ \phi_\mu(L_Q(S))\not\equiv 0$\label{alg:nbrootZ4}} 
 \STATE Compute $N_\mu$ = Algorithm \ref{alg:rad-tri-dec}$(\phi_\mu(P),\phi_\mu(Q))$\label{alg:nbrootZ5} 
   \ENDIF

\ENDFOR
\RETURN $(\mu,N_\mu)$ such that  $N_\mu$ is maximum\label{alg:nbrootZ6}
\end{algorithmic}
\end{algorithm}

\begin{lemma}\label{lem:comp-algo4}
  Algorithm \ref{algo:nbroot} computes the number of distinct solutions and a
  lucky prime for $\ideal{P,Q}$ in \comp\ bit
  operations.  Moreover, this lucky prime is upper bounded by $\sO(d^4+d^3\tau)$.
\end{lemma}

\begin{proof}
  We first prove the correctness of the algorithm.  Note first that for all
  $\mu\in B$ satisfying the constraint of Line \ref{alg:nbrootZ4},
$\phi_\mu(P)$ and $\phi_\mu(Q)$ are
  coprime. It follows that Algorithm \ref{alg:rad-tri-dec} computes the number
  of distinct solutions $N_\mu=\#V(I_\mu)$ of $I_\mu$. By
  Proposition~\ref{prop:lessroot} and Definition~\ref{def:lucky-nb-roots},
  $N_\mu\leq \#V(I)$ and the equality holds if $\mu$ is lucky for $I$.  Since
  the set $B$ of considered primes contains a lucky one by construction, the
  maximum of the computed value of $N_\mu$ is equal to $\#V(I)$.  Finally, the
  $\mu$ associated to any such maximum value of $N_\mu$ is necessarily lucky by
  the constraint of Line~\ref{alg:nbrootZ4} and since $\mu$ is larger than
  $2d^4$.

  We now prove the complexity of the algorithm.  The polynomials $P(T-SY,Y)$ and $
  Q(T-SY,Y)$  can be computed in
  $\sO_B(d^4+d^3\tau)$ bit operations by Lemma~\ref{lem:complexity:shear}.

  Proposition \ref{prop:luckyroot} states that we can compute an explicit bound
  $\Xi(d,\tau)$ in $\sO(d^4+d^3\tau)$ on the number of unlucky primes for
  $\ideal{P,Q}$.  We want to compute in Line~\ref{alg:nbrootZ2} a set $B$ of at
  least $\Xi(d,\tau)$ primes (plus one) that are larger than $2d^4$. For
  computing $B$, we can thus compute the first $\Xi(d,\tau)+2d^4+1$ prime
  numbers and reject those that are smaller than $2d^4$. The bit complexity of
  computing the $r$ first prime numbers is in $\sO(r)$ and their maximum is in
  $\sO(r)$ \cite[Theorem 18.10]{vzGGer}. We can thus compute the set of primes
  $B$ with $\sOB(d^4+d^3\tau)$ bit operations and these primes are in
  $\sO(d^4+d^3\tau)$.

Polynomials $P$, $Q$, $L_P(S)$ and $L_Q(S)$ are of degree at most $d$ in one or two variables and
they have bitsize at most $\sO(d+\tau)$ (Lemma~\ref{lem:complexity:shear}). 
The reduction of all their $O(d^2)$ coefficients modulo all the primes in $B$ can be computed via a remainder tree in a bit
  complexity that is soft linear in the total bitsize of the input
\cite[Theorem 1]{moenck1974}, which is dominated by the sum of the
    bitsizes of the $\sO(d^4+d^3\tau)$ primes in $B$ each of bitsize
     $\sO(1)$. Furthermore, computing the resultant of $\phi_\mu(P)$ and $\phi_\mu(Q)$ can be done
    with $\sO(d^3)$ arithmetic operations in $\Z_\mu$ (Lemma~\ref{complexity:subresultant}) and thus in $\sOB(d^3)$
    bit operations since
    $\mu$ has bitsize $\sO(1)$. Hence, the bit complexity of  Line~\ref{alg:nbrootZ3.1} is~$\sOB(d^4+d^3\tau)$.


Finally, the total bit complexity of
    Line~\ref{alg:nbrootZ5} is \comp, since each call to Algorithm
    \ref{alg:rad-tri-dec} has bit complexity  $\sOB(d^4)$ by Lemma~\ref{lem:cost-arith} (since $\mu$ has bitsize $\sO(1)$).
The  overall bit complexity of the algorithm is thus in \comp.
%
%
\end{proof}

\subsection{Computing a separating linear form}
\label{sec:gathering}

Using Algorithm \ref{algo:nbroot}, we now present our algorithm for computing a
linear form that separates the solutions of $\ideal{P,Q}$.

\begin{algorithm}[t]
  \caption{Separating form for $\ideal{P,Q}$}
\label{alg:sep-elem}
\begin{algorithmic}[1]
  \REQUIRE{ $P,Q$ in $\Z[X,Y]$ of total degree at most $d$ and defining a zero-dimensional  ideal $I$}
  \ENSURE{A linear form $X+aY$ that separates $V(I)$, with $a<2d^4$ and $L_P(a)\,L_Q(a)\neq 0$}
\medskip

\STATE Apply Algorithm~\ref{algo:nbroot} to compute the number of solutions $\#V(I)$
and a lucky prime $\mu$ for $I$\label{alg:sep1}

\STATE Compute $P(T-SY,Y)$, $Q(T-SY,Y)$ and
$R(T,S)=Res_Y(P(T-SY,Y),Q(T-SY,Y))$\label{alg:sep2}
\STATE Compute $R_\mu(T,S)=\phi_\mu(R(T,S))$\label{alg:sep3}
\STATE Compute $\Upsilon_\mu(S)=\phi_\mu(L_P(S))\ \phi_\mu(L_Q(S))$\label{alg:sep4}
\STATE $a:=0$\label{alg:sep5}
\REPEAT \label{alg:sep6}
\STATE Compute the degree $N_a$ of the squarefree part of $R_\mu(T,a)$\label{alg:sep7} 
\STATE $a:=a+1$ 
\UNTIL $\Upsilon_\mu(a)\neq 0$\footnotemark and $N_a=\#V(I)$ \label{alg:sep9}
\RETURN The linear form $X+aY$\label{alg:sep10}
\end{algorithmic}
\end{algorithm}

\footnotetext{$\Upsilon_\mu(S)$ is a polynomial in $\Z_\mu[S]$ and we consider  $\Upsilon_\mu(a)$ in $\Z_\mu$.}


\begin{theorem}
\label{th:algo:sep-elem} 
%
Algorithm \ref{alg:sep-elem} returns a separating linear form $X+aY$ for $\ideal{P,Q}$ with
$a<2d^4$. The bit complexity of the algorithm is in \comp.
\end{theorem}

\begin{proof}

We first prove the correctness of the algorithm. 
  We start by  proving that the value $a$ returned by the algorithm is the smallest nonnegative integer
  such that $X+aY$ separates $V(I_\mu)$ with $\Upsilon_\mu(a)\neq 0$.  
Note first that, in Line~\ref{alg:sep3}, $\phi_\mu(R(T,S))$ is indeed equal to $R_\mu(T,S)$ which is
defined as $Res_Y(P_\mu(T-SY,Y),Q_\mu(T-SY,Y))$ since the leading coefficients $L_P(S)$ and $L_Q(S)$
of $P(T-SY,Y)$ and $Q(T-SY,Y)$ do not identically vanish modulo $\mu$ (since $\mu$ is lucky), and
thus $L_{P_\mu}(S)=\phi_\mu(L_P(S))$, similarly for $Q$, and the resultant can be specialized modulo
$\mu$ \cite[Proposition 4.20]{BPR06}. 
Now, Line~\ref{alg:sep9} ensures
  that the value $a$ returned by the algorithm satisfies $\Upsilon_\mu(a)\neq 0$, and we restrict
  our attention to nonnegative such values of $a$. Note that $\Upsilon_\mu(a)\neq 0$ implies that 
$\phi_\mu(L_P(a))\ \phi_\mu(L_Q(a))\neq  0$ 
because the specialization at $S=a$ and the reduction modulo $\mu$ commute (in $\Z_\mu$).
For the same reason, $L_{P_\mu}(S)=\phi_\mu(L_P(S))$ implies $L_{P_\mu}(a)=\phi_\mu(L_P(a))$ and
thus  $L_{P_\mu}(a) \neq 0$ and, similarly, $L_{Q_\mu}(a) \neq 0$.
%
On the other hand, Line~\ref{alg:sep9} implies that the value $a$ is the
smallest that satisfies $d_T(\overline{R_\mu(T,a)})=\#V(I)$, which is also equal
to $\#V(I_\mu)$ since $\mu$ is lucky.  Lemma~\ref{lem:ineq1} thus yields that
the returned value $a$ is the smallest nonnegative integer such that $X+aY$
separates $V(I_\mu)$ and $\Upsilon_\mu(a)\neq 0$, which is our claim.

This property first implies that $a<2d^4$ because the degree of $\Upsilon_\mu$
is bounded by $2(d^2+d)$, the number of non-separating linear forms is bounded
by $d^2\choose 2$ (the maximum number of directions defined by any two of $d^2$
solutions),
and their sum is less than $2d^4$ for $d\geq 2$. Note that, since $\mu$ is
lucky, $2d^4<\mu$ and thus $a<\mu$.
The above property thus also implies, by
Proposition~\ref{prop:separation-correspondence}, that $X+aY$ separates
$V(I)$. This concludes the proof of correctness of the algorithm since $a<2d^4$
and $L_P(a)\,L_Q(a)\neq 0$ (since $\Upsilon_\mu(a)\neq 0$).

We now focus on the complexity of the algorithm.  By Lemma~\ref{lem:comp-algo4}, the bit complexity
of Line~\ref{alg:sep1} is in \comp.  The bit complexity of Lines \ref{alg:sep2} to \ref{alg:sep5} 
is in $\sOB(d^7+d^6\tau)$. 
Indeed, by Lemma~\ref{lem:complexity:shear}, $R(T,S)$ has degree $O(d^2)$ in $T$
and in $S$, bitsize $\sO(d^2+d\tau)$, and it can be computed in
$\sOB(d^7+d^6\tau)$ time.  Computing $R_\mu(T,S)=\phi_\mu(R(T,S))$ can thus be
done in reducing $O(d^4)$ integers of bitsize $\sO(d^2+d\tau)$ modulo
$\mu$. Each reduction is soft linear in the maximum of the bitsizes
\cite[Theorem 9.8]{vzGGer} thus the reduction of $R(T,S)$ can be computed in
$\sOB(d^4(d^2+d\tau))$ time {(since $\mu$ has bitsize in
  $O(\log(d^4+d^3\tau))$ by Lemma \ref{lem:comp-algo4}).}\footnote{{Note
    that $R_\mu(T,S)$ can be computed more efficiently in $\sOB(d^5+d^3\tau)$
    bit operations as the resultant of $P_\mu(T-SY,Y)$ and $Q_\mu(T-SY,Y)$
    because computing these two polynomials and their reduction can be done in
    $\sOB(d^4+d^3\tau)$ bit operations (Lemma~\ref{lem:complexity:shear}) and
    their resultant can be computed with $\sO(d^5)$ arithmetic operations in
    $\Z_\mu$ (Lemma~\ref{complexity:subresultant}) and thus with $\sOB(d^5)$ bit
    operations since $\mu$ has bitsize in $O(\log(d^4+d^3\tau))$.}}  The
computation of $\Upsilon_\mu$ can clearly be done with the same complexity since
each reduction is easier than the one in Line \ref{alg:sep3}, and the product of
the polynomials (which does not actually need to be computed since we are only
interested in whether $\Upsilon_\mu(a)$ vanishes) can be done with a bit
complexity that is soft linear in the product of the maximum degrees and maximum
bitsizes \cite[Corollary 8.27]{vzGGer}.

  We proved that the value $a$ returned by the algorithm is less than $2d^4$,
  thus the loop in Line \ref{alg:sep6} is performed at most $2d^4$ times.  Each
  iteration consists of computing the squarefree part of $R_\mu(T,a)$ which
  requires $\sOB(d^4)$ bit operations.  Indeed, computing $R_\mu(T,S)$ at $S=a$
  amounts to evaluating, {in $\Z_\mu$,} $O(d^2)$ polynomials in $S$,
  each of degree $O(d^2)$ {(by Lemma~\ref{lem:complexity:shear}). Note
    that $a$ does not need to be reduced modulo $\mu$ because $a<2d^4$ and
    $2d^4<\mu$ since $\mu$ is lucky. Thus, the bit complexity of evaluating in
    $\Z_\mu$ each of the $O(d^2)$ polynomials in $S$ is the number of arithmetic
    operations in $\Z_\mu$, which is linear the degree that is $O(d^2)$, times
    the (maximum) bit complexity of the operations in $\Z_\mu$, which is in
    $O_B(\log d\tau)$ since $\mu$ is in $\sO(d^4+d^3\tau)$ by Lemma
    \ref{lem:comp-algo4}. Hence, computing $R_\mu(T,a)$ can be done in
    $\sOB(d^4)$ bit operations.}  Once $R_\mu(T,a)$ is computed, the
  {arithmetic} complexity of computing its squarefree part in $\Z_\mu$
  is soft linear in its degree (Lemma~\ref{complexity:gcd}), that is
  {$\sO(d^2)$, which yields a bit complexity in $\sOB(d^2)$ since,
    again, $\mu$ is in $\sO(d^4+d^3\tau)$.}
%
This leads to a total bit complexity of
  $\sO_B(d^8)$ for the loop in Lines \ref{alg:sep6} to  \ref{alg:sep9}, and thus to a total bit complexity for the
  algorithm in \comp.
\end{proof}

\section{Conclusion}

{
We presented an algorithm of bit complexity $\comp$ for finding a separating linear form of a
bivariate system, improving by a factor $d^2$ the best known algorithm for this problem. Finding a
separating linear form is at the core of approaches based on rational parametrizations for solving
such systems and, as mentionned in the introduction, our algorithm directly improves the bit
complexity of the classical method for computing rational parametrizations via subresultants  \cite{VegKah:curve2d:96}. Interestingly, computing a separating linear form remains
the bit-complexity bottleneck in this algorithm  \cite{det-jsc-2009} and we show in \cite{bouzidi2013RurWorstCase} that
this is also the bottleneck for computing the rational parameterization of \cite{Rou99}. This thus yields
algorithms of bit complexity $\comp$ for computing rational parameterizations of bivariate systems
and we show in \cite{bouzidi2013RurWorstCase} that isolating boxes can be computed with a smaller
bit  complexity. It should be stressed that this complexity matches the recent one presented by
Emeliyanenko and Sagraloff \cite{sagraloff2012issacBisolve} for ``only'' computing isolating boxes
of the real solutions. Furthermore,  rational parameterizations yield efficient algorithms for 
various related problems, such as evaluating the sign of a
polynomial at the solutions of the system, or solving over-constrained systems
\cite{bouzidi2013RurWorstCase}.

One interesting open problem is to determine how, or whether, this contribution may impact the complexity
of algorithms, on plane algebraic curves,  that require finding a shear that ensures the
  curves to be in ``generic'' position (such as
  \cite{kerber11arxiv,lgv-in-etdidapc-02}). In particular, we hope  that
  this result may  improve the complexity of computing the topology of an
  algebraic plane curve.
}

\small 
\bibliographystyle{alpha}
\bibliography{paper_separating_element}

\begin{thebibliography}{LMMRS11}

\bibitem[ABRW96]{ABRW}
M.-E Alonso, E.~Becker, M.-F. Roy, and T.~W{\"o}rmann.
\newblock Multiplicities and idempotents for zerodimensional systems.
\newblock In {\em Algorithms in Algebraic Geometry and Applications}, volume
  143 of {\em Progress in Mathematics}, pages 1--20. Birkh{\"a}user, 1996.

\bibitem[BLPR13]{bouzidi2013RurWorstCase}
Y.~Bouzidi, S.~Lazard, M.~Pouget, and F.~Rouillier.
\newblock Solving bivariate systems: Efficient worst-case algorithm for
  computing rational univariate representations and applications.
\newblock INRIA Research Report 8262, 2013.

\bibitem[BPR06]{BPR06}
S.~Basu, R.~Pollack, and M.-R. Roy.
\newblock {\em Algorithms in Real Algebraic Geometry}, volume~10 of {\em
  Algorithms and Computation in Mathematics}.
\newblock Springer-Verlag, 2nd edition, 2006.

\bibitem[BSS03]{bostan2003fast}
A.~Bostan, B.~Salvy, and {\'E}.~Schost.
\newblock Fast algorithms for zero-dimensional polynomial systems using
  duality.
\newblock {\em Applicable Algebra in Engineering, Communication and Computing},
  14(4):239--272, 2003.

\bibitem[CLO97]{CLO}
D.~Cox, J.~Little, and D.~O'Shea.
\newblock {\em Ideals, Varieties, and Algorithms}.
\newblock Undergraduate Texts in Mathematics. Springer-Verlag, New York, 2nd
  edition, 1997.

\bibitem[DET09]{det-jsc-2009}
D.~I. Diochnos, I.~Z. Emiris, and E.~P. Tsigaridas.
\newblock On the asymptotic and practical complexity of solving bivariate
  systems over the reals.
\newblock {\em J. Symb. Comput.}, 44(7):818--835, 2009.

\bibitem[EK03]{Kahoui03}
M.~El~Kahoui.
\newblock An elementary approach to subresultants theory.
\newblock {\em J. Symb. Comput.}, 35(3):281--292, 2003.

\bibitem[ES12]{sagraloff2012issacBisolve}
P.~Emeliyanenko and M.~Sagraloff.
\newblock On the complexity of solving a bivariate polynomial system.
\newblock In {\em Proceedings of the 37th international symposium on Symbolic
  and algebraic computation}, ISSAC '12, 2012.

\bibitem[GLS01]{GLS01}
M.~Giusti, G.~Lecerf, and B.~Salvy.
\newblock A {G}r{\"o}bner free alternative for solving polynomial systems.
\newblock {\em J. of Complexity}, 17(1):154--211, 2001.

\bibitem[GVEK96]{VegKah:curve2d:96}
L.~Gonz{\'a}lez-Vega and M.~El~Kahoui.
\newblock An improved upper complexity bound for the topology computation of a
  real algebraic plane curve.
\newblock {\em J. Complexity}, 12(4):527--544, 1996.

\bibitem[GVN02]{lgv-in-etdidapc-02}
L.~Gonz\'alez-Vega and I.~Necula.
\newblock Efficient topology determination of implicitly defined algebraic
  plane curves.
\newblock {\em Computer Aided Geometric Design}, 19(9), 2002.

\bibitem[KS11]{kerber11arxiv}
M.~Kerber and M.~Sagraloff.
\newblock A worst-case bound for topology computation of algebraic curves.
\newblock {\em CoRR}, abs/1104.1510, 2011.

\bibitem[KS12]{kerber11JSC}
M.~Kerber and M.~Sagraloff.
\newblock A worst-case bound for topology computation of algebraic curves.
\newblock {\em J. Symb. Comput.}, 47(3):239 -- 258, 2012.

\bibitem[Laz91]{lazard1991}
D.~Lazard.
\newblock A new method for solving algebraic systems of positive dimension.
\newblock {\em Discrete Appl. Math.}, 33:147--160, October 1991.

\bibitem[LMMRS11]{Li-modpn-11}
X.~Li, M.~Moreno~Maza, R.~Rasheed, and É.~Schost.
\newblock The modpn library: Bringing fast polynomial arithmetic into maple.
\newblock {\em J. Symb. Comput.}, 46(7):841 -- 858, 2011.

\bibitem[MB74]{moenck1974}
R.~Moenck and A.~Borodin.
\newblock Fast modular transforms.
\newblock {\em Journal of Computer and System Sciences}, 8, 1974.

\bibitem[Rei97]{Reischert1997}
D.~Reischert.
\newblock Asymptotically fast computation of subresultants.
\newblock In {\em Proceedings of the 1997 international symposium on Symbolic
  and algebraic computation}, ISSAC '97, pages 233--240, New York, NY, USA,
  1997. ACM.

\bibitem[Rou99]{Rou99}
F.~Rouillier.
\newblock Solving zero-dimensional systems through the rational univariate
  representation.
\newblock {\em J. of Applicable Algebra in Engineering, Communication and
  Computing}, 9(5):433--461, 1999.

\bibitem[{Sag}12]{sagraloff2012issacNewDsc}
M.~{Sagraloff}.
\newblock {When Newton meets Descartes: A Simple and Fast Algorithm to Isolate
  the Real Roots of a Polynomial}.
\newblock In {\em Proceedings of the 37th international symposium on Symbolic
  and algebraic computation}, ISSAC '12, 2012.

\bibitem[VdW30]{VanDerWaerden}
B.~L. Van~der Waerden.
\newblock {\em Moderne Algebra I}.
\newblock Berlin, 1930.

\bibitem[vzGG99]{vzGGer}
J.~von~zur Gathen and J.~Gerhard.
\newblock {\em Modern Computer Algebra}.
\newblock Cambridge Univ. Press, Cambridge, U.K., 1st edition, 1999.

\bibitem[Yap00]{Yap-2000}
C.K. Yap.
\newblock {\em Fundamental Problems of Algorithmic Algebra}.
\newblock Oxford University Press, Oxford-New York, 2000.

\end{thebibliography}
\end{document}